\documentclass{llncs}
\def\Bbox{
{\unskip\nobreak\hfil\penalty50
\hskip1em\hbox{}\nobreak\hfil{\lower .5pt \hbox{$\Box$}}
\parfillskip=0pt \finalhyphendemerits=0 \par}
}

\def\eop{
\ifmmode {\hbox{\Bbox}} \else \Bbox \fi
}

\def\bbox{
\ifmmode {\hbox{\bbox}} \else \Bbox \fi
}

\def\rank{{\mathrm{rank}}}
\def\level{{\mathrm{level}}}
\def\firstBranch{{\mathrm{branch}}}
\def\head{{\mathrm{head}}}

\def\maxNodes{{\mathrm{maxNodes}}}
\def\fr{{\mathrm{fr}}}
\def\lfr{{\mathrm{lfr}}}
\def\X{{\mathcal{X}}}
\def\L{{\mathcal{L}}}

\usepackage[utf8x]{inputenc}
\usepackage{amsmath, amssymb}
\usepackage{algorithmicx}
\usepackage{algpseudocode}

\title{Operational characterization of scattered MCFLs\\Technical Report}
\author{Zolt\'an \'Esik \and Szabolcs Iv\'an}
\institute{University of Szeged, Hungary}

\begin{document}

\maketitle

\begin{abstract}
We give a Kleene-type operational characterization of
Muller context-free languages (MCFLs) of well-ordered 
and scattered words.
\end{abstract}

\section{Introduction}

A word, called `arrangement' in \cite{Courcelle}, is an isomorphism type of a 
countable labeled linear order. They form a generalization of the classic
notions of finite and $\omega$-words.

Finite automata on $\omega$-words have by now a vast literature, see \cite{PerrinPin}
for a comprehensive treatment. Finite automata acting on well-ordered words longer 
than $\omega$ have been investigated in   
\cite{Bedon96,Buchi73,Choueka,Wojciechowski84,Wojciechowski85}, to mention
a few references.
In the last decade, the theory of automata on well-ordered words 
has been extended to automata on all countable words, including scattered and dense words.
In \cite{Bedonetal,BesCarton,Bruyereetal}, both operational
and logical characterizations of the class of languages of countable words 
recognized by finite automata were obtained.

Context-free grammars generating $\omega$-words were introduced in 
\cite{CohenGold} and subsequently studied in \cite{Boasson,Nivat}.  
Context-free grammars generating arbitrary countable 
words were defined in \cite{EsikIvanBuchi,EsikIvanMuller}. Actually, two types of grammars 
were defined, context-free grammars with B\"uchi acceptance condition (BCFG),
and context-free grammars with Muller acceptance condition (MCFG). 
These grammars generate the B\"uchi and the Muller context-free languages 
of countable words, abbreviated as BCFLs and MCFLs. 
Every BCFL is clearly an MCFL, but there exists an MCFL of well-ordered words 
that is not a BCFL, for example the set of all countable well-ordered words 
over some alphabet. In fact, it was shown in \cite{EsikIvanBuchi}
that for every BCFL $L$ of well-ordered words there is an integer $n$ such that 
the order type of the underlying linear order of every word in $L$ is bounded 
by $\omega^n$.

A Kleene-type characterization of BCFLs of well-ordered and
scattered words was given in \cite{EsikOkawa}. 
Here we provide a Kleene-type characterization of MCFLs of well-ordered 
and scattered words. 
Before presenting the necessary preliminaries in detail, 
we give a formulation of our main result, at least in the well-ordered case. 

Suppose that $\Sigma$ is an alphabet, and let $\Sigma^\sharp$ denote the set of 
all (countable) words over $\Sigma$.  Let $P(\Sigma^\sharp)$ be the set 
of all subsets of $\Sigma^\sharp$. The set of $\mu\omega T_w$-expressions over
 $\Sigma$ is defined by the following grammar:
\begin{eqnarray*}
T &::=& a\ |\ \varepsilon\ |\ x\ |\ T+T\ |\ T\cdot T\ |\ \mu x . T\ |\ T^\omega
\end{eqnarray*}
Here, each letter $a \in \Sigma$ denotes the language containing $a$ as its unique 
word, while $\varepsilon$ denotes the language containing only the empty word. 
The symbols $+$ and $\cdot$ are interpreted as set union and concatenation 
over $P(\Sigma^\sharp)$, and the variables $x$ range over languages
 in  $\Sigma^\sharp$. The $\mu$-operator corresponds to taking least fixed points.
Finally, $^\omega$ is interpreted as the $\omega$-power operation 
over $P(\Sigma^\sharp)$: $L \mapsto L\cdot L \cdots$\ . An expression is closed if 
each variable occurs in the scope of a least fixed-point operator. 
Each closed expression denotes a language in $P(\Sigma^\sharp)$.
Our main result in the well-ordered case, which is a corollary of Theorem~\ref{thm-scattered} is:

\begin{theorem}
\label{thm-well ordered}
A language $L \subseteq \Sigma^\sharp$ is an MCFL of well-ordered words iff 
it is denoted by some closed $\mu\omega T_w$-expression.
\end{theorem} 

\begin{example}
The expression $\mu x. (x^\omega + a + b + \varepsilon)$ denotes the set of all 
well-ordered words over the alphabet $\{a,b\}$. 
\end{example} 

It was shown in \cite{EsikOkawa} that the syntactic fragment of the above expressions,
with the $\omega$-power operation restricted to closed expressions, characterizes
the BCFLs of well-ordered words. 
A similar, but more involved result holds for MCFLs of scattered words,
cf. Theorem~\ref{thm-scattered}. 
Both theorems were conjectured by the authors of \cite{EsikOkawa}.

\section{Notation}

\subsection{Linear orderings}
  A \emph{linear ordering} is a pair $(I,<)$, where $I$ is a set and $<$ is an irreflexive transitive trichotomous relation (i.e. a strict total ordering) on $I$.
  If $I$ is finite or countable, we say that the ordering is finite or countable as well.
  \emph{In this paper, all orderings are assumed to be countable.}
  A good reference for linear orderings is \cite{Rosenstein}.
  
  An \emph{embedding} of the linear ordering $(I,<)$ into $(J,\prec)$ is an order preserving function $f:I\to J$, i.e. $x<y$ implies $f(x)\prec f(y)$ for
  each $x,y\in I$. If $f$ is surjective, we call it an \emph{isomorphism}. Two linear orderings are said to be \emph{isomorphic} if there exists an isomorphism
  between them. Isomorphism between linear orderings is an equivalence relation;
  classes of this equivalence relation are called \emph{order types}.
  If $I\subseteq J$ and $<$ is the restriction of $\prec$ onto $I$, then we say that $(I,<)$ is a 
  \emph{sub-ordering} of $(J,\prec)$.
  
  Examples of linear orderings are the ordering $(\mathbb{N},<)$ of the positive integers, the ordering $(\mathbb{N}_{-},<)$ of the negative integers,
  the ordering $(\mathbb{Z},<)$ of the integers
  and the ordering $(\mathbb{Q},<)$ of the rationals. The respective order types are denoted $\omega$, $-\omega$, $\zeta$ and $\eta$.
  In order to ease notation, we write simply $I$ for $(I,<)$ if the ordering $<$ is standard or known from the context.
  
  An ordering is \emph{scattered} if it does not have a sub-ordering of order type $\eta$, 
  otherwise it is \emph{quasi-dense}.
  An ordering is a \emph{well-ordering} if it does not have a sub-ordering of order type $-\omega$. Order types of well-orderings are called \emph{ordinals}.
  
  When $(I,<)$ is an ordering and for each $i\in I$, $(J_i,<_i)$ is an ordering, then the \emph{generalized sum}
  $\mathop\sum\limits_{i\in I}(J_i,<_i)$ is the disjoint union $\{(i,j):i\in I, j\in J_i\}$ equipped with the 
  lexicographic ordering $(i,j)<(i',j')$ iff $i<i'$, or $i=i'$ and $j<_ij'$.
  It is known that if $(I,<)$ and the $(J_i,<_i)$ are scattered or well-ordered, then so is the generalized sum. 
  The operation of generalized sum can be extended to order types since it preserves isomorphisms.
  For example, $\zeta=-\omega+\omega$.
  Ordinals are also equipped with an exponentiation operator. 
  
Hausdorff classified linear orderings into an infinite hierarchy.
Following \cite{Khoussainovetal}, we present a variant of this hierarchy. 
Let $VD_0$ be the collection of all finite linear orderings,
  and when $\alpha$ is some ordinal, let $VD_\alpha$ be the collection of all finite sums of linear orderings of the form
  $\mathop\sum\limits_{i\in\mathbb{Z}}(I_i,<_i)$, where for each integer $i\in\mathbb{Z}$, $(I_i,<_i)$ is a member of  $VD_{\alpha_i}$
  for some ordinal $\alpha_i<\alpha$. According to a theorem of Hausdorff (see e.g.~\cite{Rosenstein}, Thm. 5.24), a (countable) linear ordering $(I,<)$ is scattered if and only if it belongs to $VD_\alpha$
  for some (countable) ordinal $\alpha$; the least such $\alpha$ is called the \emph{rank} of $(I,<)$, denoted $\mathrm{rank}(I,<)$.
    
\subsection{Words, tree domains, trees}
  An \emph{alphabet} is a finite nonempty set $\Sigma$ of symbols, usually called \emph{letters}.
  A \emph{word} over $\Sigma$ is a linear ordering $(I,<)$ equipped with a \emph{labeling function} $\lambda:I\to\Sigma$. 
  An \emph{embedding of words} is a mapping preserving the order and the labeling; a surjective embedding is an \emph{isomorphism}. 
  Order theoretic properties of the underlying linear ordering of a word 
   are transferred to the word.
  A word is finite if its underlying linear order is finite,
  and an $\omega$-word, if its underlying linear order is a well-order of order type 
  $\omega$. 
  We usually identify isomorphic words and denote by $\Sigma^\sharp$ the set of all words over $\Sigma$. 
  As usual, we denote the collection of finite and $\omega$-words over $\Sigma$ by $\Sigma^*$ and 
  $\Sigma^\omega$, respectively.  The length of a word
   $u \in \Sigma^*$ is denoted $|u|$.
  A language over $\Sigma$ is a subset of $\Sigma^\sharp$. As in the introduction, 
  we let $P(\Sigma^\sharp)$ denote the
  collection of all languages over $\Sigma$.

  When $(I,<)$ is a linear ordering and  $w_i=(J_i,<_i,\lambda_i)$ for $i \in I$ are words, 
  then we define their \emph{concatenation} $\mathop\prod_{i\in I}w_i$ as the word with 
  underlying linear order 
   $\mathop\sum\limits_{i\in I}(J_i,<_i)$ and labeling $\lambda(i,j)=\lambda_i(j)$. 
  When $I$ has two elements, we obtain the usual notion of concatenation, denoted $u\cdot v$, or just $uv$.  The operation of concatenation is 
  extended to languages in $P(\Sigma^\sharp)$: $\mathop\prod_{i\in I}L_i = \{
  \mathop\prod_{i\in I}w_i : w_i \in L_i\}$. When $L,L_1,L_2 \subseteq \Sigma^\sharp$,
  then we define $L_1 + L_2$ to be the set union and $L_1L_2 = \{uv : u \in L_1,\ v \in L_2\}$. 
  Moreover,  we define $L^\omega = \mathop\prod_{i \in \mathbb{N}} L$. 
  
  The set $P(\Sigma^\sharp)$ of languages over $\Sigma$, equipped with the inclusion order, 
  is a complete lattice.
  When $A$ is a set, a function $f:P(A)^n\to P(A)$ is \emph{monotone} if $A_i\subseteq A_i'$ for each $i\in[n]$ implies 
  $f(A_1,\ldots,A_n)\subseteq f(A_1',\ldots,A_n')$.
  The following fact is clear. 
  
  \begin{lemma}
  \label{lem-cont1}
  The functions $+,\cdot : P(\Sigma^\sharp)^2 \to P(\Sigma^\sharp)$ and $^\omega:  P(\Sigma^\sharp) \to P(\Sigma^\sharp)$ are monotone.
  \end{lemma}
  
  We will also consider \emph{pairs} of words over an alphabet $\Sigma$,
  equipped with a finite concatenation and an $\omega$-product operation.
  For pairs $(u,v)$, $(u',v')$ in $\Sigma^\sharp\times\Sigma^\sharp$, we define the
  product $(u,v)\cdot (u',v')$ to be the pair $(uu',v'v)$, and when
  for each $i\in\mathbb{N}$, $(u_i,v_i)$ is in $\Sigma^\sharp\times\Sigma^\sharp$, then we let $\mathop\prod\limits_{i\in\mathbb{N}}(u_i,v_i)$ be
  the word $\bigr(\mathop\prod\limits_{i\in\mathbb{N}}u_i\bigr)\bigl(\mathop\prod\limits_{i\in\mathbb{N}_{-}}v_i\bigr)$. Let $P(\Sigma^\sharp\times \Sigma^\sharp)$ denote the set of all subsets 
  of $\Sigma^\sharp \times \Sigma^\sharp$. Then $P(\Sigma^\sharp\times \Sigma^\sharp)$ is naturally 
  equipped with the operations of set union $L + L'$, concatenation $L\cdot L' = 
  \{(u,v)\cdot (u',v'): (u,v) \in L,\ (u',v') \in L'\}$ and Kleene star 
  $L^* = \{\varepsilon\} \cup L \cup L^2 \cup \cdots$. We also define an $\omega$-power operation 
  $P(\Sigma^\sharp \times \Sigma^\sharp) \to P(\Sigma^\sharp)$ by 
  $L^\omega = \{ \mathop\prod\limits_{i\in\mathbb{N}}(u_i,v_i) : (u_i,v_i) \in L \}$.
  When $L_1,L_2\subseteq \Sigma^\sharp$, let $L_1 \times L_2= \{(u,v) : u \in L_1,\ v \in L_2\}
  \subseteq \Sigma^\sharp \times \Sigma^\sharp$. 
  
  \begin{lemma}
  \label{lem-cont2}
  The functions
  \begin{eqnarray*}
  &&\times: P(\Sigma^\sharp)^2 \to  P(\Sigma^\sharp \times \Sigma^\sharp)\\
  && +,\cdot: P(\Sigma^\sharp \times \Sigma^\sharp)^2
  \to P(\Sigma^\sharp \times \Sigma^\sharp)\\
  && ^*: P(\Sigma^\sharp \times \Sigma^\sharp)
  \to P(\Sigma^\sharp \times \Sigma^\sharp)\\
  && ^\omega: P(\Sigma^\sharp \times \Sigma^\sharp)
  \to P(\Sigma^\sharp)
  \end{eqnarray*}
  are monotone.
  \end{lemma}
  
  We will use Lemma~\ref{lem-cont1} and Lemma~\ref{lem-cont2} in the following context. Suppose that 
  for each $i \in [n]  = \{1,\ldots,n\}$, 
  $f_i : P(\Sigma^\sharp)^{n + p} \to P(\Sigma^\sharp)$ is a function that can be constructed 
  by function composition from the above functions, the projection functions and constant 
  functions. Let $f = \langle f_1,\ldots,f_n\rangle : 
  P(\Sigma^\sharp)^{n + p} \to P(\Sigma^\sharp)^n$ be the target tupling of the $f_i$.
  Then $f$ is a monotone function, and by Tarski's fixed point theorem, for each $y \in P(\Sigma^\sharp)^p$
  there is a least solution of the fixed point equation $x = f(x,y)$ in the variable $x$ ranging
  over $P(\Sigma^\sharp)^n$. This least fixed point, denoted $\mu x. f(x,y)$, gives rise to a 
  function $P(\Sigma^\sharp)^p \to P(\Sigma^\sharp)^n$ in the parameter $y$. It is known that 
  this function is also monotone, see e.g. \cite{BEbook}.


  A \emph{tree domain} is a prefix closed nonempty (but possibly infinite) subset of $\mathbb{N}^*$.
  Elements of a tree domain $T$ are also called nodes of $T$.
  When $x$ and $x\!\cdot\! i$ are nodes of $T$ for $x\in\mathbb{N}^*$ 
  and $i\in\mathbb{N}$, then $x\!\cdot\! i$ is a \emph{child} of $x$.
  A \emph{descendant} of a node $x$ is a node of the form $x\!\cdot\!y$, where $y\in\mathbb{N}^*$.
  Nodes of $T$ having no child are the \emph{leaves} of $T$. 
  The leaves, equipped with order inherited from the lexicographic ordering of $\mathbb{N}^*$ 
  form the \emph{frontier} of $T$, denoted $\fr(T)$. 
  An \emph{inner node} of $T$ is a non-leaf node.
  Subsets of a tree domain $T$ which themselves are tree domains are called \emph{prefixes} of $T$.
  A \emph{path} of a tree domain $T$ is a prefix of $T$ such that each node has at most one child.
  A path can be identified with the unique sequence $w$ in $\mathbb{N}^{\leq\omega}$ of all sequences over $\mathbb{N}$ of length at most $\omega$ 
  such that the set of nodes of the path consists of the finite prefixes of $w$.
  A path $\pi$ of $T$ is \emph{maximal} if no path of $T$ contains $\pi$ properly.
  When $T$ is a tree domain and $x\in T$ is a node of $T$, 
  then the \emph{sub-tree domain} $T|_x$ of $T$ is the set $\{y:xy\in T\}$.
  A tree domain $T$ is \emph{locally finite} if each node has a descendant which is a leaf. 
  
  A \emph{tree} over an alphabet $\Delta$ is a mapping $t:\mathrm{dom}(t)\to\Delta \cup\{\varepsilon\}$, where $\mathrm{dom}(t)$ is a tree domain,
  such that inner vertices are mapped to letters in $\Delta$.
  Notions such as nodes, paths etc. of tree domains are lifted to trees. 
  When $\pi$ is a path of the tree $t$, then $\mathrm{labels}(\pi) = \{ t(u) : u \in \pi \}$ is the set of labels of 
  the nodes of $\pi$, and  
  $\mathrm{infLabels}(\pi)$ is the set of labels occurring infinitely often. 
  For a path $\pi$, 
  $\mathrm{head}(\pi)$ denotes the minimal node $x$ of $\pi$ (with respect to the prefix order) with $\mathrm{infLabels}(\pi)=\mathrm{labels}(\pi|_x)$,
  if $\pi$ is infinite; otherwise $\mathrm{head}(\pi)$ is the last node of $\pi$.
  The labeled \emph{frontier word} $\mathrm{lfr}(t)$ of a tree $t$ is determined by the leaves \emph{not} labeled by $\varepsilon$, which is equipped
  with the lexicographic ordering of $\mathbb{N}^*$ and the labeling function of $t$.
  It is worth observing that when $\pi=x_0,x_1,\ldots$ is an infinite path of a tree $t$ and for each $i$, $\alpha_i$ ($\beta_i$, resp.)
  is the word determined by the leaf labels of the descendants of $x_i$ to the left (right, resp.) of $x_{i+1}$ (i.e.
  if $x_{i+1}$ is the $j$th child of $x_i$, then $\alpha_i=\mathrm{lfr}(t|_{x\!\cdot\!1})\cdot\mathrm{lfr}(t|_{x\!\cdot\! 2})\cdot\ldots\cdot\mathrm{lfr}(t|_{x\!\cdot\!(j-1)})$ 
  and similarly for $\beta_i$), then $\mathrm{lfr}(t)=\mathop\prod\limits_{i\in\mathbb{N}}(\alpha_i,\beta_i)$.
  
\subsection{Muller context-free languages of scattered words}

  A \emph{Muller context-free grammar}, or MCFG for short, is a system $G=(V,\Sigma,R,S,\mathcal{F})$, where $V$ is the alphabet of nonterminals,
  $\Sigma$ is the alphabet of terminals, $\Sigma \cap V = \emptyset$, $R$ is the finite set of productions of the form $A\to\alpha$ with $A\in V$ and $\alpha\in(\Sigma\cup V)^*$,
  $S\in V$ is the start symbol and $\mathcal{F}\subseteq P(V)$ is the set of 
  \emph{nonempty} accepting sets.
  
  A \emph{derivation tree} of the above grammar $G$ is a tree $t:\mathrm{dom}(t)\to V\cup\Sigma\cup\{\varepsilon\}$ satisfying the following conditions:
  \begin{enumerate}
  \item For each inner node $x$ of $t$ there exists a rule $X\to X_1\ldots X_n$ in $R$ such that $t(x)=X$, the children of $x$ are exactly
    $x\cdot 1,\ldots, x\cdot n$, and for each $i\in [n]$, $t(x\cdot i)=X_i$ so that when $n=0$, $x$ has a single
    child $x\cdot 1$ labeled $\varepsilon$;
  \item For each infinite path $\pi$ of $t$, $\mathrm{infLabels}(\pi)$ is an accepting set of $G$.
  \end{enumerate}
  
  A derivation tree is \emph{complete} if its leaves are all labeled in $\Sigma\cup\{\varepsilon\}$.
  If $t$ is a derivation tree having \emph{root symbol} $t(\varepsilon) = A$, then we say that 
  $t$ is an $A$-tree.
  The language $L(G,A) \subseteq \Sigma^\sharp$ \emph{generated} from $A \in V$ is the 
  set of frontier words of complete $A$-trees. The language $L(G)$ generated by $G$ is $L(G,S)$. 
  An MCFL is a language generated by some MCFG. 
  
  \begin{example}
  \label{expl-grammar-wellordered}
  If $G=(\{S,I\},\{a,b\},R,S,\{\{I\}\})$, with $$R=\{S\to a,S \to b,S \to \varepsilon, S\to I,I\to SI\},$$ then $L(G)$ consists of all the well-ordered words over $\{a,b\}$. 
  \end{example}
  
  \begin{example}
  \label{expl-grammar-scattered}
  If $G=(\{S,I\},\{a,b\},R,S,\{\{I\}\})$, with \[R= \{S\to a,S \to b,S \to \varepsilon, S\to I,I\to SIS\},\] then $L(G)$ consists of all the scattered words over $\{a,b\}$. 
  \end{example}

  Let $L\subseteq\Sigma^\sharp$ be an MCFL consisting of scattered words only and $G=(V,\Sigma,R,S,\mathcal{F})$ an
  MCFG with $L(G)=L$. We may assume that $G$ is in \emph{normal form} 
  \cite{EsikIvanMuller} -- among the
  properties of this normal form we will use the following ones (see \cite{EsikIvanMuller}, Prop. 14) frequently:
  \begin{itemize}
  \item For every derivation tree there is a locally finite 
        derivation tree with the same root symbol and same labeled frontier.
  \item The frontier of each derivation tree is scattered. 

  \end{itemize}
  
  {\bf In the rest of the paper, we fix an MCFG $G=(V,\Sigma,R,S,\mathcal{F})$ 
  in normal form generating only scattered words.}
  
  When $t$ is a derivation tree, then we define $\rank(t) = \rank(\fr(t))$.
  For a derivation tree $t$, let $\mathrm{maxNodes}(t)$ be the prefix of 
  $\mathrm{dom}(t)$ consisting of the nodes having maximal rank,
  i.e. $\mathrm{maxNodes}(t)=\{x\in\mathrm{dom}(t):\mathrm{rank}(t|_x)=\mathrm{rank}(t)\}$.
  Suppose that $t$ is locally finite. It is known, (see e.g. ~\cite{EsikIvanLatin}, proof of Proposition 1, paragraph 4) 
  that in this case $\mathrm{maxNodes}(t)$ is the union of finitely many maximal paths. 
  Clearly, the set $\{\pi_1,\ldots,\pi_n\}$ of these paths is unique. 
  Let $\mathrm{level}(t)$ stand for the above $n$, the number of 
  maximal paths covering $\mathrm{maxNodes}(t)$.
  Also, let $\firstBranch(t)$ stand for the longest common prefix
  of the paths $\pi_1,\ldots,\pi_n$ (which is a finite word if 
  $\mathrm{level}(t)>1$ and is $\pi_1$ if $\mathrm{level}(t)=1$).

  We say that a (not necessarily locally finite) 
  derivation tree $t$ is \emph{simple} if $\maxNodes(t)$ 
  contains a single infinite path $\pi$ and if 
  $\mathrm{infLabels}(\pi)=\mathrm{labels}(\pi)$, i.e. $\mathrm{head}(\pi)=\varepsilon$.
  (When $t$ is additionally locally finite, then this path $\pi$ contains all 
  nodes of $\maxNodes(t)$.) Such a path is called the \emph{central path} of $t$. 
  If $t$ is a simple $A$-tree and $F$ is the set of labels of its central path, 
  then we call $t$ an \emph{$F$-simple $A$-tree}.

\section{The main result}

For locally finite complete derivation trees $t'$ and $t$, let $t'\prec t$ if one of the following conditions holds:
  \begin{enumerate}
  \item $\mathrm{rank}(t')<\mathrm{rank}(t)$;
  \item $\mathrm{rank}(t')=\mathrm{rank}(t)$ and $\mathrm{level}(t')<\mathrm{level}(t)$;
  \item $\mathrm{rank}(t')=\mathrm{rank}(t)$, $\mathrm{level}(t')=\mathrm{level}(t)>1$ and
    $|\firstBranch(t')|<|\firstBranch(t)|$.
  \item $\mathrm{rank}(t')=\mathrm{rank}(t)$, $\mathrm{level}(t')=\mathrm{level}(t)=1$, that is, the
    set of nodes of maximal rank is a path $\pi$ in $t$ and a path $\pi'$ in $t'$. Then let $t'\prec t$ iff
    $|\mathrm{head}(\pi')|<|\mathrm{head}(\pi)|$.
  \end{enumerate}

  \begin{lemma}
  \label{lem-prec}
  The relation $\prec$ is a well-partial order (wpo) of \emph{locally finite}
  complete derivation trees. The minimal 
  elements of this wpo are the one-node trees corresponding to the elements of $\Sigma \cup\{\varepsilon\}$. 
  Suppose 
  that $t$ is a locally finite complete derivation
  tree and $t'=t|_x$ is a proper subtree of $t$, so that $x\neq\varepsilon$.
  If $t$ is not simple, or if $t$ is simple but $x$ does not belong to the central path of $t$, 
  then $t'\prec t$.
  \end{lemma}

  \begin{proof}
  It is clear that $\prec$ is irreflexive. To prove that it is transitive, suppose that 
  $t''\prec t'$ and $t' \prec t$. If $\rank(t'') < \rank(t)$, then clearly $t'' \prec t$.
  Suppose that $\rank(t'') = \rank(t)$. Then also $\rank(t'') = \rank(t') = \rank(t)$. 
  If $\level(t'') < \level(t)$ then $t'' \prec t$ again. Thus, we may suppose that 
  $\level(t'') = \level(t)$, so that $\level(t'') = \level(t') = \level(t) = n$.
  Now there are two cases. If $n > 1$, then, since $t'' \prec t'$ and $t' \prec t$,
  we know that $|\firstBranch(t'')| < |\firstBranch(t')| < |\firstBranch(t)|$ and thus 
  $t'' \prec t$. If $n = 1$, then the maximal nodes form a single maximal 
  path in each of the trees $t'',t'$ and $t$.
  Let us denote these paths by $\pi'',\pi'$ and $\pi$, respectively. As $t''\prec t'$ and $t'\prec t$,
  we have that $|\head(\pi'')| < |\head(\pi')| < |\head(\pi)|$, so that $t'' \prec t$ again. 
  
  The fact that there is no infinite decreasing sequence of locally finite complete derivation 
  trees with respect to the relation $\prec$ is clear, since every set of 
  ordinals is well-ordered. 
  
  Suppose now that $t$ is a locally finite complete derivation tree which has at least two nodes.
  By assumption, $t$ has a leaf node $x$. Let $t' = t|_x$. If $\rank(t') <  \rank(t)$ 
  then $t' \prec t$. Otherwise, $\rank(t') = \rank(t) = 0$ and $t$ is necessarily finite 
  (since the frontier of an infinite complete derivation tree is infinite).
  Clearly, $\maxNodes(t)$ is the set of all nodes of $t$, 
  and either $\level(t') = 1 < \level(t)$, or $\level(t') = \level(t) = 1$. 
  In the latter case,
  $t$ has a single maximal path $\pi$, and $|\head(\pi')| = 0 < |\head(\pi)|$ 
  for the single maximal path $\pi'$ of $t'$. In either case, $t' \prec t$. 
  Thus, no locally finite complete derivation tree having more than one node is minimal. 
  On the other hand, all one-node complete derivation trees corresponding to 
  the elements of $\Sigma \cup\{\varepsilon\}$ are clearly minimal (and locally finite).
  
  To prove the last claim, suppose that $t$ is a locally finite complete derivation tree 
  and $t'=t|_x$. If $\mathrm{rank}(t')<\mathrm{rank}(t)$, we are done.
  Otherwise, $\mathrm{rank}(t')=\mathrm{rank}(t)$ and $x$ is a member of $\mathrm{maxNodes}(t)$.
  Thus, if $\pi$ is a maximal path of $\mathrm{maxNodes}(t')$, then $x\pi$ is a maximal path of $\mathrm{maxNodes}(t)$.
  Hence $\mathrm{level}(t')\leq \mathrm{level}(t)$. If $\mathrm{level}(t')<\mathrm{level}(t)$, we are done.
  Otherwise, $\mathrm{level}(t')=\mathrm{level}(t)$ and $\mathrm{maxNodes}(t)=x\mathrm{maxNodes}(t')$.
  

  Now there are two  cases.
  \begin{enumerate}
  \item If $\mathrm{level}(t)>1$, then $\firstBranch(t)=x\firstBranch(t')$, thus
     $|\firstBranch(t')|<|\firstBranch(t)|$ and $t' \prec t$.
  \item Suppose that $\level(t) = 1$, and let $\pi$ denote the unique maximal path of  
      $t$ whose nodes form the set $\mathrm{maxNodes}(t)$. Since $\rank(t') = \rank(t)$,
      we have that $x$ belongs to $\pi$ and, by assumption, $t$ is not simple. 
      Since $t$ is not simple and has at least two nodes, 
      $\head(\pi) \neq \varepsilon$ and $|\head(\pi')|  < |\head(\pi)|$, 
      where $\pi'$ is the unique maximal path of $t'$ whose nodes form the set $\mathrm{maxNodes}(t')$.
      (Actually $\pi'$ is determined by the proper suffix $\pi|_x$ of $\pi$.)
   \eop
  \end{enumerate}
  
  \end{proof}
  
  Now we define certain ordinary $\omega$-regular languages \cite{Muller,PerrinPin}
  corresponding to central paths of simple derivation trees.
  Let $\Gamma$ stand for the (finite) set consisting of those triplets
  \[(\alpha,B,\beta)\in (V\cup\Sigma)^*\times V\times(V\cup\Sigma)^*\]
  for which $\alpha B\beta$ occurs as the right-hand side of a production of $G$.
  For any nonterminal $A\in V$ and accepting set $F\in\mathcal{F}$, let $R_{A,F}\subseteq\Gamma^\omega$
  stand for the set of $\omega$-words over $\Gamma$ accepted by the deterministic (partial) Muller (word)
  automaton $(F,\Gamma,\delta,A,\{F\})$, with $B=\delta(C,(\alpha,D,\beta))$ if and only if
  $D=B$ and $C\to\alpha B\beta$ is a production of $G$.
  By definition, each $R_{A,F}$ is an $\omega$-regular set which can 
  be built from singleton sets corresponding to the elements of $\Gamma$ by the 
  usual regular operations and the $\omega$-power operation (actually,
  since every state has to be visited infinitely many times, $R_{A,F}$ can be written as the
  $\omega$-power of a regular language of finite words over $\Gamma$). 
  
  Members of $R_{A,F}$ correspond to central paths of $F$-simple $A$-trees
  in the following sense. Given $w=(\alpha_1,A_1,\beta_1)(\alpha_2,A_2,\beta_2)\ldots\in R_{A,F}$,
  we define an $F$-simple $A$-tree $t_w$ of $G$ as follows.
  The nodes $x_0,x_1,\ldots$ of the central path of $t_w$ are $x_0=\varepsilon$,
  and  $x_i=x_{i-1}\cdot(|\alpha_i|+1)$, for $i > 0$. Each $x_i$ has $|\alpha_{i+1}A_{i+1}\beta_{i+1}|$
  children, respectively labeled by the letters of the word $\alpha_{i+1}A_{i+1}\beta_{i+1}$.
  Nodes not on the central path of $t_w$ are leaf nodes.
  
  It is straightforward to see the following claims:
  \begin{enumerate}
  \item For each $w\in R_{A,F}$, $t_w$ is an $F$-simple $A$-tree.
  \item Every $F$-simple $A$-tree has a prefix of the form $t_w$, for some $w\in R_{A,F}$.
    Thus, every such tree can be constructed by choosing an appropriate $w\in R_{A,F}$, and substituting
    a derivation tree $t_x$ with root symbol $t_w(x)$ for each leaf $x$ of $t_w$.
  \end{enumerate}
  
  Moreover, it is clear that when  $w=(\alpha_1,A_1,\beta_1)(\alpha_2,A_2,\beta_2)\ldots$, then
  $\lfr(t_w)$ is $(\prod_{i \in \mathbb{N}} \alpha_i)\cdot (\prod_{i \in \mathbb{N}_-} \beta_i)$.

Let us assign a variable $X_A$ to each $A \in V$,
and let $\mathcal{X}$ be the set of all variables. 
For each ordinary regular expression $r$ over $\Gamma$,
we define an expression (term)  $\overline{r}$ over $\Sigma \cup \X$
involving the function symbols $\times,
+,\cdot$. To this end, when $\alpha$ is a 
word in $(\Sigma \cup V)^*$, let $\overline{\alpha}$ 
be the word in $(\X \cup \Sigma)^*$ obtained by replacing each occurrence
of a nonterminal $A$ 
by the variable $X_A$. Then, for a letter 
$\gamma = (\alpha,A,\beta)\in \Gamma$, define $\overline{\gamma}
= \overline{\alpha} \times \overline {\beta}$. To obtain
$\overline{r}$, we replace each occurrence of a letter $\gamma$ 
in $r$ by $\overline{\gamma}$.

When $A$ is a nonterminal and $A \in F$ for some $F \in \mathcal{F}$,
consider an ordinary regular expression $r_{A,F}$ over $\Gamma$
such that $r_{A,F}^\omega$ denotes the set $R_{A,F}$ (defined above) of all 
$\omega$-words corresponding to 
central paths of $F$-simple $A$-trees.
Then consider the following system of equations $E_G$ 
associated with $G$ in the variables $\mathcal{X}$:

\[X_A=
\mathop\sum\limits_{A\to u \in R}\overline{u} \ +\ 
\mathop\sum\limits_{A\in F\in \mathcal{F}}(\overline{r_{A,F}})^\omega.
\]

\begin{example}The system of equations $E_G$ associated with the grammar 
in Example~\ref{expl-grammar-scattered} is:
\begin{eqnarray*}
X_S &=& a + b + \varepsilon + X_I\\
X_I &=& (X_S\times X_S)^\omega
\end{eqnarray*}
\end{example}

As usual, we can associate a function $f_G: P(\Sigma^\sharp)^{\mathcal{X}} \to P(\Sigma^\sharp)^{\mathcal{X}}$ 
with $E_G$.  
By Lemmas~\ref{lem-cont1} and \ref{lem-cont2} and using the facts that the projections 
are monotone and that monotone functions are closed under function composition,
we have that $f_G$ is monotone. Thus, $f_G$ has a least fixed point. 

\begin{proposition}
\label{prop-fixed point}
For each $A \in V$, the corresponding component of the
least fixed point solution of the system $E_G$ is 
the language $L(G,A)$ of all words derivable from $A$.  
\end{proposition}

\begin{proof}
The fact that the languages $L(G,A),\ A \in V$, form a solution is clear from the 
definition of $E_G$. Let us also define $L(G,a) = \{a\}$, for each $a \in \Sigma \cup\{\varepsilon\}$.
Suppose that the family of languages $L_A,\ A \in V$ is another 
solution, and let $L_a = \{a\}$ for $a \in \Sigma \cup\{\varepsilon\}$.
We want to show that if $t$ is a locally finite complete $A$-tree with $\lfr(t) = u$,
 then $u \in L_A$, for each $A \in \Sigma \cup\{\varepsilon\} \cup V$.
We apply well-founded induction with respect to the wpo $\prec$.

For the base case, if $t$ consists of a single node, then $A = a\in \Sigma \cup\{\varepsilon\}$, $u = a$,  
and our claim is clear. Otherwise, there are two cases: either $t$ is a simple tree, or not.

If $t=A(t_1,\ldots,t_n)$ is not simple, then we have $t_i\prec t$ for each $i\in[n]$ 
by Lemma~\ref{lem-prec}. Let $A_i$ be the root symbol of $t_i$ and $u_i$ the 
labeled frontier word of $t_i$ for each $i$.
By the induction hypothesis, each $u_i$ is a member of $L_{A_i}$. Since $t$ is a derivation tree, 
$A\to A_1\ldots A_n$ is a production of $G$. Thus, by the construction of $E_G$, $u = u_1\ldots u_n \in L_A$.

Otherwise, if $t$ is an $F$-simple $A$-tree for some $F\in\mathcal{F}$ and $A \in V$, then $t$ 
can be constructed from a tree $t_w$ with $w\in R_{A,F}$ by
replacing each leaf node $x$ of $t_w$ by some complete derivation tree $t_x$ with root symbol
$t_w(x)$. Since such leaves are not on the central path of $t$, we have $t_x\prec t$
for each $x$, again by Lemma~\ref{lem-prec}. Applying the induction hypothesis, we get that
the labeled frontier word $u_x$ of each $t_x$ is a member of $L_{t_w(x)}$. Thus, by the 
construction of $E_G$, $u$ is a member of $L_A$.
\eop 
\end{proof}


It is well-known, cf.  \cite{Bekic,DeBakkerScott} or \cite{BEbook}, Chapter 8, Theorem 2.15 and Chapter 6, Section 8.1, Equation (3.2),
that when $\L,\L',\L''$ are complete lattices and 
$f: \L \times \L' \times \L'' \to \L$ and $g: \L \times \L' \times \L'' \to \L'$
are monotone functions, then the least solution (in the parameter $z$) 
of the system of equations 
\begin{eqnarray*}
x &=& f(x,y,z)\\
y &=& g(x,y,z)
\end{eqnarray*}
can be obtained by Gaussian elimination as
\begin{eqnarray*}
x &=& \mu x. f(x, \mu y. g(x,y,z),z)\\
y &=& \mu y. g( \mu x. f(x, \mu y. g(x,y,z),z)  ,y,z)
\end{eqnarray*}
Using this fact and Proposition~\ref{prop-fixed point}, we obtain our final result.

Let the set of $\mu\omega T_s$-expressions over the alphabet $\Sigma$ be defined by the following grammar
(with $T$ being the initial nonterminal):
\begin{eqnarray*}
T &::=& a\ |\ \varepsilon\ |\ x\ |\ T+T\ |\ T\cdot T\ |\ \mu x . T\ |\ P^\omega\\
P &::=& T\times T\ |\ P+P\ |\ P\cdot P\ |\ P^*
\end{eqnarray*}
Here, $a\in \Sigma$ and $x\in\mathcal{X}$ for an infinite countable set of variables.
An occurrence of a variable is \emph{free} if it is not in the scope of a $\mu$-operation,
and bound, if it is not free. A \emph{closed expression} does not have free variable
occurrences.  
The semantics of these expressions are defined as expected using the monotone functions 
over $P(\Sigma^\sharp)$ and $P(\Sigma^\sharp \times \Sigma^\sharp)$ introduced earlier. 
When the free variables of an expression form the set $\mathcal{Y}$, then an expression 
denotes a language in $P((\Sigma \cup \mathcal{Y})^\sharp)$.

\begin{remark}
Actually, $\varepsilon$ is redundant, as it is expressible by $((\mu x.x \times \mu x.x)^*)^\omega$. 
We do not need a constant $0$ denoting the empty set of pairs since it is expressible by 
$(\mu x.x) \times (\mu x.x)$.
\end{remark}

\begin{theorem}
\label{thm-scattered}
A language $L\subseteq\Sigma^\sharp$ is an MCFL of scattered words if and only if it can be denoted by a closed $\mu\omega T_s$-expression.
\end{theorem}

\begin{proof}
It is easy to show that each expression denotes an MCFL of scattered words.
One uses the following facts, where $\Delta$ denotes an alphabet and $x,\#\not\in \Delta$.
\begin{enumerate}
\item The set of MCFLs (of scattered words) over $\Delta$ is closed under $+$ and $\cdot$.
\item If $L,L'\subseteq \Delta^\sharp$ are MCFLs (of scattered words), then $L\# L' \subseteq (\Delta \cup\{\#\})^\sharp$ is an MCFL (of scattered words).
\item Suppose that $L,L'\subseteq \Delta^\sharp\#\Delta^\sharp$ are MCFLs (of scattered words).
Then $$\{uv\#v'u' : u\#u'\in L,\ v\#v'\in L'\} \subseteq \Delta^\sharp\#\Delta^\sharp$$ 
is an MCFL (of scattered words). 
\item Suppose that $L\subseteq \Delta^\sharp\#\Delta^\sharp$ is an MCFL (of scattered words). 
Then $$\{u_1\ldots u_n\#v_n\ldots v_1 : n\geq 0,\ u_i\#v_i \in L\} \subseteq \Delta^\sharp\#\Delta^\sharp$$ is an MCFL (of scattered words). 
\item Suppose that $L\subseteq \Delta^\sharp\#\Delta^\sharp$ is an MCFL (of scattered words). 
Then $$\{(u_1u_2\ldots)(\ldots v_2v_1) :  u_i\#v_i \in L\} \subseteq \Delta^\sharp$$ is an MCFL (of scattered words).
\item Suppose that $L \subseteq (\Delta \cup\{x\})^\sharp$ is an MCFL (of scattered words). Then, with respect to set inclusion, there is a
least language $L' \subseteq\Delta^\sharp$ such that $L[x \mapsto L'] = L'$,
and this language $L'$ is an MCFL (of scattered words). (Here, $L[x \mapsto L']$ is the language obtained from $L$ 
by `substituting' $L'$ for $x$.)
\end{enumerate}
It is known (see \cite{EsikIvanMuller}) that the class of MCFLs is (effectively) closed under substitution and that
every context-free language of finite words (in particular, $\{a,b\}$, $\{ab\}$ or $\{a\#b\}$) is an MCFL,
showing Items 1--3 above.

For Items 4 and 5, let $G=(V,\Delta\cup\{\#\},R,S,\mathcal{F})$ be an MCFG generating the MCFL $L\subseteq\Delta^\sharp\#\Delta^\sharp$.
Then
\[ G_1=(V\cup\{\#\},\Delta\cup\{\#'\},R\cup\{\#\to\#',\#\to S\}, \#,\mathcal{F})\]
generates the MCFL $L_1=\{u_1\ldots u_n\#'v_n\ldots v_1:n\geq 0, u_i\#v_i\in L\}$, showing Item 4 (applying the substitution $\#'\mapsto \{\#\}$) and
\[ G_2=(V\cup\{\#\},\Delta,R\cup\{\#\to S\}, \#,\mathcal{F}\cup\{H\cup\{\#\}:H\subseteq V\})\]
generates the MCFL defined in Item 5.

Finally, let $G=(V,\Delta\cup\{x\},R,S,\mathcal{F})$ be an MCFG generating $L\subseteq(\Delta\cup\{x\})^\sharp$. Then
\[ G_3=(V\cup\{x\},\Delta, R\cup\{x\to S\}, x,\mathcal{F})\]
generates the language $L'$ of Item 6.

The other direction follows from Proposition~\ref{prop-fixed point}.
\eop
\end{proof}

\begin{example}
The expression $\mu x. ((x \times x)^\omega + a + b + \varepsilon)$ denotes the set of all scattered words 
over the alphabet $\{a,b\}$.  
\end{example}

\begin{example}
Let $L \subseteq\{a,b\}^\sharp$ be the language of all words $w$ such that the 
word obtained from $w$ by removing all occurrences of letter $b$ is well-ordered,
as is the `mirror image' of the word obtained by removing all occurrences of letter $a$. 
It is not difficult to show that each word in $L$ 
contains only a finite number of `alternations' between $a$ and $b$. 
Using this fact, an MCFG generating $L$ is:
 $G=(\{S,A,B,I,J\},\Sigma,R,S,\{\{I\},\{J\}\})$ with $R$ consisting of the productions
\begin{eqnarray*}
S &\to& AS\ |\ BS\ |\ \varepsilon\\
A &\to& a\ |\ \varepsilon\ |\ I\\
I &\to& AI\\
B &\to& b\ |\ \varepsilon\ |\ J\\
J &\to& JB
\end{eqnarray*}
\end{example}
 Using the algorithm described above (with some simplification), an expression for $L$ is:
\begin{eqnarray*}
t_S &=& \mu x_S.\bigl((t_A+t_B)x_S+\varepsilon\bigr)
\end{eqnarray*}
with 
\begin{eqnarray*}
t_A &=& \mu x_A.\bigl( a+\varepsilon+ (x_A \times \varepsilon)^\omega \bigr)\\
t_B &=& \mu x_B.\bigl( b+\varepsilon+ (\varepsilon \times x_B)^\omega \bigr).
\end{eqnarray*} 

We restate Theorem~\ref{thm-well ordered} and show that it is a corollary of Theorem~\ref{thm-scattered}.

{\bf Theorem.}
A language $L \subseteq \Sigma^\sharp$ is an MCFL of well-ordered words iff 
it is denoted by some closed $\mu\omega T_w$-expression.

\begin{proof}
Recall that the set of $\mu\omega T_w$-expressions over an alphabet $\Sigma$ is defined by the grammar
\begin{eqnarray*}
T &::=& a\ |\ \varepsilon\ |\ x\ |\ T+T\ |\ T\cdot T\ |\ \mu x . T\ |\ T^\omega
\end{eqnarray*}
where $a\in\Sigma$ and $x$ ranges over the set $\mathcal{X}$ of variables, moreover, an expression $t$ is closed
if each occurrence of a variable $x$ in $t$ is within the scope of some prefix $\mu x.$
Below we will sometimes view the construct $t^\omega$ as a shorthand for $(t\times\varepsilon)^\omega$.

For one direction, we show by structural induction that for a $\mu\omega T_w$-expression $t$ 
with free variables in $X$, the language $|t| \subseteq (\Sigma\cup X)^\sharp$ denoted by $t$ consists of well-ordered 
words. For the base cases, i.e. when $t=a$, $t=\varepsilon$ or $t=x$, the claim clearly holds.
If $t = t_1 + t_2$ or $t = t_1 \cdot t_2$, or $t = t_1^\omega$, for some expressions $t_1,t_2$, 
our claim is again clear (using the fact that every well-ordered product of well-ordered words is
well-ordered in the last two cases). 
Finally, if $t=\mu x.t_1$, where $t_1$ denotes an MCFL $L\subseteq (\Sigma\cup X)^\sharp$, 
$|t|$ is the language 
$\mathop\bigcup\limits_{\alpha\geq 0}L_\alpha$, where $L_0=\emptyset$ and for each $\alpha>0$, 
$$L_{\alpha}= L_{<\alpha} \cup L[x\mapsto L_{<\alpha}]$$
where $L_{<\alpha}=\left(\mathop\bigcup\limits_{\beta<\alpha}L_\beta\right)$.
Thus, if $L$ contains only well-ordered words then so does each $L_\alpha$, since 
languages of well-ordered words are closed under substitution.

For the other direction, we may restrict ourselves to expressions (of type $T$ or $P$) 
which do not have any subexpression denoting the empty set, nor any subexpression 
other than $\epsilon$ denoting $\{\varepsilon\}$. 

Suppose that $t$ and $p$ are such expressions of type $T$ and $P$, respectively.
It is not difficult to prove the following claim by (simultaneous) structural induction: 

{\em Claim A.} If $t$ has a subexpression (belonging to the syntactic category $P$)
of the form $t_1\times t_2$ with $t_2\neq\varepsilon$, then $|t|$ contains a word 
which is not well-ordered.
If $p$ has a subexpression $p'=t_1\times t_2$ with $t_2\neq \varepsilon$, 
then $|p|$ contains a pair $(u,v)$ such that either $v \neq \varepsilon$
or one of $u,v$ is not well-ordered. 

To prove this, first note that $t$ cannot have the form $a$, $\varepsilon$ or $x$.
When $p=t_1\times t_2$, for some $t_1,t_2$, $t_2 \neq \varepsilon$,
then our claim clearly holds for $p$, since either one of $|t_1|$ 
and $|t_2|$ contains a word which is not well-ordered, or $|t_2|$ 
contains a nonempty word.
The induction step is clear when $p = p_1 + p_2$, $p = p_1 \cdot p_2$, $p = p_1^*$, 
or when $t = t_1 + t_2$, $t= t_1 \cdot t_2$, or $t = p_1^\omega$. 
When $t = \mu x.t_1$, then $t_1$ contains a word $u$ which is not well-ordered.
Since by assumption $|t|$ contains a nonempty word $v$, $t$ contains 
$u[x \mapsto v]$, which is not well-ordered.

To complete the proof, note that if each subexpression of $t$ of the form $t_1\times t_2$ 
satisfies $t_2=\varepsilon$, then we can transform $t$ into an equivalent $\mu \omega T_w$ 
expression by repeatedly replacing subexpressions of the form $t_1 \times \varepsilon$
with $t_1$ and subexpressions of the form $t_1^*$ with $\mu x. (t_1x + \varepsilon).$ 
\eop
\end{proof}

\def\isEmpty#1{{\textsc{isEmpty}(#1)}}
\def\Symbols#1{{\textsc{Symbols}(#1)}}
\def\True{{\textsc{True}}}
\def\False{{\textsc{False}}}
\def\Ordo{{\mathcal{O}}}

Using Claim A, we may develop a low-degree polynomial-time algorithm for the following
decision problem: given a closed $\mu\omega T_s$-expression $t$ of syntactic category $T$, does 
the language denoted by $t$ consist of well-ordered words only?
The expression $t$ may be assumed to be given as an expression tree.

In the following, $t_1,t_2$ denote expressions belonging to the syntactic category $T$ and 
$p_1,p_2$ denote expressions of syntactic category $P$. Expressions $e,e_1,e_2$ are arbitrary.
We also allow the symbol $\emptyset$ to appear in expressions, which denotes the empty language.
 
In the first step of the algorithm, we transform $t$ into an equivalent expression 
$t_\emptyset$ which is either the symbol $\emptyset$, or contains no subexpression 
denoting the empty set. This can be done by a straightforward algorithm 
in linear time using the fact that an expression of the form $\mu x. t_1$ denotes the 
empty language iff $t_1[x/\emptyset]$, the expression obtained from $t_1$ by replacing 
each free occurrence of $x$ in $t_1$ by $\emptyset$ denotes the empty language.

Suppose now that $t_\emptyset$ is not the symbol $\emptyset$, so that $t_\emptyset$ is 
not empty. We construct another equivalent expression in which each subexpression of 
syntactic category $T$ denoting $\{\varepsilon\}$
is $\varepsilon$ itself. To achieve this, we determine for each subexpression $e$ of $t_\emptyset$ 
the set $\Symbols{e}\subseteq \Sigma\cup\mathcal{X}$
containing all the symbols that occur in some word of $|e|$ 
(or in a word in a pair of $|e|$, if $e$ is of type $P$).
The recursion rules for this are: 
\begin{align*}
\Symbols{\varepsilon} = \emptyset,\ \Symbols{x} =&\ \{x\} ,\ \Symbols{a} =\{a\},\\
\Symbols{e_1e_2} = \Symbols{e_1+e_2} =&\ \Symbols{e_1}\cup\Symbols{e_2},\\
\Symbols{p_1^*}=\Symbols{p_1^\omega} =&\ \Symbols{p_1},\\
\Symbols{t_1\times t_2} =&\ \Symbols{t_1}\cup\Symbols{t_2},\\
\Symbols{\mu x.t_1} =&\ \Symbols{t_1}-\{x\}.
\end{align*}
Note that the correctness of these rules (e.g. the one for concatenation) 
depends on the assumption that no subexpression of $t_\emptyset$  denotes the empty set.

Having computed $\Symbols{e}$ for each subexpression $e$, observe that 
$|e| =\{\varepsilon\}$ for a subexpression $e$ of syntactic category $T$ 
if and only if $\Symbols{e}=\emptyset$.
Hence, during the computation of $\Symbols{.}$, we can flag each subexpression of $t_\emptyset$ of type
$T$ by a bit indicating whether it denotes the language $\{\varepsilon\}$.
Using this information, we can then replace each maximal subexpression denoting 
$\{\varepsilon\}$ by $\varepsilon$, yielding an equivalent expression
$t_{\emptyset\varepsilon}$ containing no occurrence of the symbol $\emptyset$ 
such that each subexpression of type $T$ different from $\varepsilon$
denotes a language containing at least one nonempty word.
 Applying now Claim A to $t_{\emptyset\varepsilon}$, we get the desired decision procedure
 answering the question whether the given closed expression $t$ denotes a language of 
 well-ordered words. 

All steps can be performed in (deterministic) linear time 
in the usual RAM model of computation, say, 
except for the computation of the function
$\Symbols{.}$ whose time complexity depends on the data structure chosen for representing 
sets of symbols. If this data structure is a self-balancing binary tree, which supports 
the construction of $\emptyset$ and the singleton sets in constant time,
the removal of one element from an $n$-element set in $\Ordo(\log n)$ time
and the construction of the union of two sets with $n$ and $k$ elements 
in $\Ordo(\min\{n,k\}\cdot\log(n+k))$ time
(destroying the two sets, which is not a problem since only their emptiness flag is needed later, 
which is already stored), respectively,
then we get an overall time complexity of $\Ordo(n\cdot \log^2 n)$. Thus we have shown the following:
\begin{corollary}
The problem whether an arbitrary closed $\mu\omega T_s$-expression of syntactic category $T$
denotes a language which 
consists of well-ordered words only,
can be decided in $\Ordo(n\cdot \log^2 n)$ time (in the usual RAM model of computation).
\end{corollary}

\end{document}